\documentclass{llncs}
\usepackage{amssymb, mymacros}
\usepackage{graphicx}

\usepackage{url}
\urldef{\mailsa}\path|{alfred.hofmann, brigitte.apfel, ursula.barth, christine.guenther,|
\urldef{\mailsb}\path|ingrid.haas, frank.holzwarth, anna.kramer, leonie.kunz, nicole.sator,|
\urldef{\mailsc}\path|erika.siebert-cole, peter.strasser, lncs}@springer.com|
\newcommand{\keywords}[1]{\par\addvspace\baselineskip
\noindent\keywordname\enspace\ignorespaces#1}

\newtheorem{axiom}{\bfseries Axiom}
\newtheorem{assum}{\bfseries Assumption}
\newcommand{\wsq}{\qquad $\square$}

\begin{document}
%%%%%%%%%%%%%%%%%%%%%%%%%%%%%%%%%%%%%%%%%%%%%%%%%%%%%%%%
% Title
%%%%%%%%%%%%%%%%%%%%%%%%%%%%%%%%%%%%%%%%%%%%%%%%%%%%%%%%
%\title{A value based on marginal contributions \\ 
%for multi-alternative games with restricted coalitions}
%\author{Satoshi Masuya and Masahiro Inuiguchi}

\mainmatter

\title{%Cooperative Games with Unknown Coalitional Values
A Partially Defined Game with Costs }

%\titlerunning{%Cooperative Games with Unknown Coalitional Values
%Toward the Theory of Cooperative Games under Incomplete %Coalitional Value 
%Information
%}

\author{Satoshi Masuya
%\thanks{Please note that the LNCS Editorial assumes that all authors have used
%the western naming convention, with given names preceding surnames. This determines
%the structure of the names in the running heads and the author index.}%
}
%
%\authorrunning{Satoshi Masuya and Masahiro Inuiguchi}


\institute{Faculty of Business Administration, Daito Bunka University.\\
1-9-1, Takashimadaira Itabashi-ku TOKYO 175-8571 (Japan)\\
\email{masuya@ic.daito.ac.jp}
}

\maketitle
\begin{abstract}
The present study explores a problem that can be resolved by employing the notion of a partially defined cooperative game, 
yet cannot by using a restricted game.  
The following situation is considered:
First, it is assumed that the worth of the grand and singleton coalitions are known.
It takes some amount of costs to obtain worth of unknown coalitions. 
If it is performed, then the worth of the grand coalition is decreased by the value of a cost function.
With the view point of fairness of a payoff allocation, we should examine coalitional worth
as much as possible. However, we should stop examining coalitional worth at some point since total payoff
is reduced by continuing the examinations.
We name the new decision making problem a partially defined cooperative game with costs.
The problem of a partially defined cooperative game with costs is finding the solution of partially defined cooperative games 
at each point and the best exiting rule of examinations of coalitional worth.
%We consider the situation described below:
%first, we assume that the worth of the grand and singleton coalitions is only known. 
%It take some amount of costs to obtain worth of larger coalitions. 
%If it is performed, then players make a payment from the worth of the grand coalition. 
%That is, the worth of the grand coalition is reduced by examinations of coalitional worth. 
\end{abstract} 

AMS Classification Code: 91A12

Competing interests: The authors declare no competing interests.

\section{Introduction}

Cooperative game theory provides a useful tool to analyze various cost and/or surplus allocation problems, the distribution of voting power in a parliament, and so on. 
This theory is employed to analyze problems that involve $n$ entities called players which are usually expressed by characteristic
functions that map each subset of players to a real number. The solutions are given by a set of $n$-dimensional real numbers or a value function that assigns a real number to each player. 
Such a real number can represent the cost borne by the player, the power of influence, allocation of shared profits, and so on. % depending on the problem setting.

We deal with a cooperative game that is called a partially defined cooperative game in this paper. 
A partially defined cooperative game (a PDG, in short) is a cooperative game in which the worth of some coalitions are unknown. 
A cooperative game is called full if worth of all coalitions are known.

Willson~\cite{willson93} first considered partially defined cooperative games. 
He introduced a generalized Shapley value \cite{shapley53} derived solely from the known worth of coalitions in a game.
%He proposed a generalized Shapley value \cite{shapley53} that is obtained by using only the known worth of coalitions of a game.
%However, the generalized Shapley value coincides with the ordinary Shapley value of a game whose coalitional worth is zero 
%if it is unknown and given values otherwise, which seems to be not well justified.  

After that, many results are obtained in this area. 
Recently, Yu \cite{yu21} studied a cooperative game with a coalition structure under limited feasible coalitions. 
Here, a PDG with a coalition structure was considered, and he developed and axiomatized its Owen value \cite{owen77}.
Aguilera et al. \cite{aguilera10} and Calvo and Gutiérrez \cite{calvo15} proposed 
independently the same extension of the Shapley value for PDGs from different 
points of view. This extended value was 
characterized for games with restricted cooperation in Albizuri et al. \cite{masuya22} with three axioms which are more elementary than that proposed in \cite{calvo15}.
That is, this type of extended Shapley value has been often studied.
On the other hand, \v{C}ern\'{y} and Grabisch \cite{cerny24} studies so called player-centered PDGs and 
they derived the collection of monotonic full games which can be obtained from a PDG within the class of such games.
%they obtained the set of monotonic full games that 
%can be obtained from a PDG on the class of such games. 

Partially defined games and restricted games are mathematically equivalent, despite the fact that 
the underlying reasons for the presence of unknown coalitions are different.
The worth of coalitions in PDGs are unknown due to the prohibitive cost of examining all possible coalitions. 
In restricted games, the worth of coalitions are unknown since those coalitions cannot form because of 
the difference of ideologies among players, and so on.

Myerson \cite{myerson77} first considered restricted games using the set of feasible coalitions
that is called communication situations. Subsequently, he proposed and axiomatized the Shapley value for restricted games 
which is called the Myerson value. 
In the line of this research, many studies generalize the set of feasible coalitions of a restricted game.
As representative studies, conference structures by Myerson \cite{myerson80} and union stable systems by Algaba et al.~\cite{algaba01}
can be mentioned. 
 
The present study explores a problem that can be resolved by employing the notion of a PDG, yet cannot by using a restricted game.

%In this paper, we investigate a new problem which can be solved by using a concept of a PDG but cannot 
%by using that of a restricted game.  

The following situation is considered:
First, it is assumed that the worth of the grand and singleton coalitions are known.
It takes some amount of costs to obtain worth of unknown coalitions. 
If it is performed, then the worth of the grand coalition is decreased by the value of a cost function. 
%For instance, if the examination is performed for coalitions of size two, 
%then the coalitional worth of $N$ is reduced to $v(N) - c_v(2)$. 
%For the sake of simplicity,
%we assume $c(s) = b\cdot s$ such that $b \ge 0$ in this study. 

With the view point of fairness of a profit allocation, we should examine coalitional worth
as much as possible. However, we should stop examining coalitional worth at some point since total payoff
is reduced by continuing the examinations.   
We name the new decision making problem a partially defined game with costs.
The problem of a PDG with costs is finding the solution of PDGs at each point and 
the best exiting rule of examinations of coalitional worth.

We extend the Shapley value for PDGs that is proposed in Aguilera et al. \cite{aguilera10}, Calvo and Gutiérrez \cite{calvo15},
and Albizuri et al.\cite{masuya22} to PDGs with costs and axiomatize the proposed value. Furthermore, we propose a rule 
to finish examinations of coalitional worth for such problems and axiomatize the proposed rule. 

This paper is organized as follows. 
In Section 2, we present the definitions of partially defined games and their Shapley value.
In Section 3, we define a partially defined game with costs and the assumptions a cost function has in this paper. 
In Section 4, we propose and axiomatize the Shapley value for PDGs with costs. 
In Section 5, we propose a rule to stop examinations of coalitional worth and 
axiomatize the proposed rule. 
In Section 6, concluding remarks are given.

\section{The Shapley Value and Partially Defined Games}

In this section, we provide the definition and related concepts of partially defined games and 
the Shapley value for PDGs that is proposed by Albizuri et al.~\cite{masuya22}, which is used to define the Shapley value
for PDGs with costs. 

Let $N=\{1,2,\ldots,n \}$ be the set of players. A non-empty set of players $S \subseteq N$ is called 
a coalition. Then the pair $(N,v)$ where $v:2^N \to \R $ is called a TU-game. For every coalition $S \subseteq N$, 
$v(S)$ is called the worth of the coalition $S$.  
In this paper, we assume a property for TU-games as follows:
\begin{equation}\label{eq:2-1}
v(S) \ge 0 \Forall S \subseteq N, \text{ and } v(N) > 0,
\end{equation}
The set of TU-games which satisfy the condition above is denoted $\Gamma^N$.

%A game $(N,v)$ is superadditive if and only if:
%\begin{equation}\label{eq:superadd}
%v(S \cup T) \ge v(S)+v(T),\ 
%\forall S,\ T \subseteq N \text{\ such that\ } 
%S \cap T=\emptyset.
%\end{equation}
%Superadditivity
%is a natural property that gives each player an incentive to form a larger coalition.
%In this paper, we deal with the superadditive games whose coalitional worth is nonnegative. 
%Hence, the set of superadditive games which satisfy $v(S) \ge 0$ for all $S \subseteq N$ is denoted $\Gamma^N$.

The Shapley value \cite{shapley53} is a well-known one-point solution concept for TU-games and 
its explicit form is:
\begin{equation}
\phi_i(v)=\sum_{ \stackrel{\scriptstyle S \subseteq N }{S \ni i}  }\frac{(|S|-1)!(n-|S|)!}{n!} (v(S)-v(S\setminus i )),\ \forall i \in N.
\end{equation}

A PDG,
on the set of known coalitions  ${\cal K}$ is a set-function $v$
which maps every set $S \in {\cal K}$ a real number $v(S)$, such that $v(\emptyset) = 0$. 
A triple $(N,{\cal K}, v)$ identifies a PDG. 
Usually, it is assumed that the worth of grand coalition is known. 
The set of PDGs with the set of known coalitions ${\cal K}$ is denoted $\Gamma^{N, {\cal K}}$. 

%Moreover, we assume that $v$ is superadditive in the following sense:
%\begin{eqnarray}\label{eq:2-1}
%v(S) &\geq& \sum_{i=1}^s v(T_i),\ \forall S, T_i 
%\in {\cal K}, \ i=1,2,\ldots,s
%\mbox{ such that }
%\bigcup_{i=1,2,\ldots,s} T_i 
%= S \nonumber\\ 
%& & \mbox{ and $T_i$, $i=1,2,\ldots,s$ are disjoint.} 
%\end{eqnarray}

We provide the definition of an extension of the Shapley value for PDGs by Albizuri et al.~\cite{masuya22}. 
According to Harsanyi's procedure \cite{harsanyi63}, all the members of a coalition $S$ receive
a dividend from $S$. We will distinguish two cases depending on whether the coalitional worth is known
or not: (i) if $S \in {\cal K}$ the total amount of the dividends allocated by all the subcoalitions
of $S$ is $v(S)$, (ii) otherwise, i.e. if  $S \not\in {\cal K}$, the dividend of $S$ is zero.
Formally the procedure can be described as follows. If $v$ is a PDG, define recursively a function
$d_v: 2^N \to \R $ by:
\begin{equation}
d_v(\emptyset) = 0; \text{ and} 
\end{equation}
\begin{equation}
d_v(S) =\left\{
\begin{array}{cc}
v(S) -\sum_{T \subsetneq S} d_v(T), & \text{ if } S \in  {\cal K}, \\
0, & \text { if } S \not\in {\cal K}.
\end{array}
\right.
\end{equation}

The real number $\frac{d_v(S)}{|S|}$
is usually called the Harsanyi dividend of coalition $S$ in $v$.
The Shapley value for PDGs by \cite{aguilera10}\cite{calvo15}\cite{masuya22} $\Hat{\phi}$ is defined as follows:
\begin{equation}
\Hat{\phi}_i(v) = \sum_{S \subseteq N, S \ni i} \frac{d(v,S)}{|S|} \Forall i \in N.
\end{equation}

\section{Partially Defined Games with Costs}

In this section, we define a PDG with costs and related concepts.

%Next, we define a partially defined game with payments. 
The following situation is considered:
On the first stage, we assume that the worth of the grand and singleton coalitions are only known. 
It takes some amount of costs to obtain worth of unknown coalitions. The cost function of $v \in \Gamma^N$ 
is defined by $c_v: 2^N  \setminus {\cal L} \to \R_+$ where ${\cal L} = \{ \{ 1 \}, \dots, \{ n \}, N \}$ and $\R_+= \{ r \in \R | r \ge 0\}$.
If it is performed, then the worth of the grand coalition is decreased by the value of a cost function. 
For instance, if the examination is performed for $S \subset N$, then the coalitional worth of $N$ is reduced to $v(N) - c_v(S)$. 
%For the sake of simplicity,
%we assume $c(s) = b\cdot s$ such that $b \ge 0$ in this study. 
If $c_v(S)=0$ for all $S \in 2^N  \setminus {\cal L}$, then all coalitional worth can be 
obtained with no costs. With the view point of fairness of an allocation of payoffs, we should examine coalitional worth
as much as possible. However, we should stop examining coalitional worth at some point since total payoff
is reduced by continuing the examinations.   

In this situation, consider a third party who implements a game.   
If he is not satisfied with the allocated payoffs of the first game (that is, the game
in which the worth of the grand and singleton coalitions are only known) and agrees to the payment to examine the 
worth of coalition $S$ whose cost is lowest, that is, $c_v(S) \le c_v(R) \Forall R \in 2^N \setminus {\cal L}$, 
then the worth of the coalition $S$ is obtained. 
Then the same dialog is performed. If he is not satisfied with the allocated payoffs even if he considers
the cost, then the worth of a coalition $T$ whose cost is the second lowest, that is, $c_v(T) \ge c_v(S)  \text{ and } c_v(T) \le c_v(R) \Forall R \in 2^N \setminus ( {\cal L} \cup \{ S \})$,
is obtained. 
%If he obtained the worth of coalitions whose cardinality is $n-1$, then the examination is finished and players are allocated  
%the payoffs of the last game. 
In this study, we assume that the third party knows the order of values  of unknown coalitional worth although 
he does not know the values  of them. That is, he knows the following relation:
\begin{equation}
0 \le c_v(S_1^v) \le c_v(S_2^v) \le, \dots, \le c_v(S_{2^n-n-2}^v),
\end{equation}
where $\{ S_1^v, S_2^v, \dots, S_{2^n-n-2}^v \}$ is the set of unknown coalitions in $v$. 
If there is no confusion, then $S_k^v$ is denoted $S_k$.

To follow the rule above, he examine the worth of unknown coalitions with the order 
$S_1^v, S_2^v, \dots, S_{2^n-n-2}^v$. If $c_v(S_l^v) = c_v(S_{l+1}^v)$, then we assume that the worth
of coalitions $S_l^v$ is decided to be examined prior to $S_{l+1}^v$ in some manner.  
This order is denoted ${\cal S}^v = (S_1^v, S_2^v, \dots, S_{2^n-n-2}^v)$ for all $v \in \Gamma^N$.

A triple $(N, v, c_v)$ identifies a PDG with costs.
%We assume that a game satisfies superadditivity and $v(S) \ge 0$ for all $S \subseteq N$ such that $|S| \le k$ or $|S| = n$. 
Since $N$ is fixed in this study, we simply write $v$ and $c_v$ to represent a game and a cost function if there is no confusion.
%Similarly, we simply write $c(s)$ instead of $c(s,v)$ if there is no confusion. 
The problem of a PDG with costs is finding the solution of PDGs and the best exiting rule of examinations of coalitional worth.

Here, we give several assumptions which are related to a cost function in this study. 
Let $v, w \in \Gamma^{N}$. Then we define $(v + w)(S) = v(S) + w(S)$
for every $S \subseteq N$. 
Let $x \vee y$ represents $\max \{ x, y \}$ and $x \wedge y$ represents $\min \{ x, y \}$
where $x,y \in \R$. 

%\begin{assum}\label{as:1}
%Values of a cost function are shared by all players if they are known.
%\end{assum}

%\begin{assum}\label{as:2}
%Let $v \in \Gamma^N$ and $k \in \{2, \dots, n-1 \}$. If the examination is proceeded to stage $k$, then the following holds:
%\begin{equation}
%v(N)-\sum_{l=2} ^k c_v(l) \ge \sum_{i=1}^r v(S_i) \text{ for all } \bigcup_{i=1}^{r} S_i = N \text{ such that } |S_i| \le k \Forall i \in \{1, \dots, r \}    
%\end{equation}
%are disjoint. 
%\end{assum}

%Assumption \ref{as:2} states that examinations of coalitional worth are performed if and only if a game satisfies superaddditibity.  

\begin{assum}\label{as:4}
Let $v \in \Gamma^N$ and $S \subset N$ such that $|S| \ge 2$. Then, $c_v(S)=0$ holds 
if $v(S) = 0$.
\end{assum}

%If the examination of coalitional worth yields $v(S) = 0$ for all $S \subset N$ with $|S| = k$, 
%then employing the solution $\Hat{\phi}$ (for example) provides no information to players. 
%Thus, Assumption \ref{as:4}  states that $c_v(s) = 0$ for such a result.

%\begin{assum}\label{as:5}
%Once examinations of coalitional worth are stopped, then there is no possibility to resume it. 
%\end{assum}

%\begin{assum}\label{as:6}
%Let $v, w \in \Gamma^N$. Then the following holds:

%\begin{equation}
%c_{v \vee w} (S) = c_{v} (S) \vee c_{w} (S) \text{ for all } S \subset N \text{ such that } |S| \ge 2. 
%\end{equation} 
%\begin{equation}
%c_{v \wedge w} (S) = c_{v} (S) \wedge c_{w} (S) \text{ for all } S \subset N \text{ such that } |S| \ge 2. 
%\end{equation} 
%\end{assum}

%\begin{assum}\label{as:7}
%Let $v \in \Gamma^N$. If $v$ is monotonic, then the following holds:

%\begin{equation}
%c_{v} (S) \ge c_{v} (T) \text{ for all } S \supseteq T \text{ such that } |S|, |T| \ge 2. 
%\end{equation} 
%\end{assum}

%Assumption \ref{as:7} states that 
%if a game $v$ is monotonic, then $c_v$ is also monotonic. 

%\begin{assum}[Additivity] \label{as:8}
%Let $v,w \in \Gamma^N$. Then the following holds:
%\begin{equation}
%c_{v+w} (S) = c_{v} (S) + c_{w} (S) \text{ for all } S \subset N \text{ such that } |S| \ge 2. 
%\end{equation} 
%\begin{equation}
%c_{v-w} (S) = c_{v} (S) - c_{w} (S) \text{ for all } S \subset N \text{ such that } |S| \ge 2. 
%\end{equation} 
%\end{assum}

%Assumption \ref{as:8} states that a cost function satisfies Additivity. 

\begin{assum}\label{as:9}
Let $v \in \Gamma^N$. We know the ordering relation:
\begin{equation}
0 \le c_v(S_1^v) \le c_v(S_2^v) \le \dots \le c_v(S_{2^n-n-2}^v)
\end{equation}
where $S_k^v \in 2^N\setminus$ ${\cal L}$ for all $k \in \{ 1, \dots, 2^n-n-2 \}$.
\end{assum}

%\begin{assum}\label{as:10}
%Let $v \in \Gamma^N$. Then the following holds:

%\begin{equation}
%c_{v} (S) = c_{-v} (S) \text{ for all } S \in 2^N \setminus {\cal K}
%\end{equation} 
%\end{assum}

%Assumption \ref{as:10} states that 
%if absolute values of coalitional worth are same between two games,
%then the costs of them are the same.  

\begin{assum}\label{as:11}
Let $v, w \in \Gamma^N$. Then the following holds:
\begin{equation}
c_{v+w}(S_{k}^{v+w})=c_v(S_k^v) + c_w(S_k^w),
\end{equation}
\begin{equation}
c_{v-w}(S_{k}^{v-w})=c_v(S_k^v) - c_w(S_k^w),  
\end{equation}
for all $k \in \{ 1,2, \dots, 2^n-n-2 \}$.
\end{assum}
%%%%%%%%%%%------------koko(20251119)-------------------------------------
%In this study, we consider the order of coalitions that is the same as the order of 
%its component games. 
In this study, we assume that the cost function of the sum game $v+w$ is the sum of 
the component games $v$ and $w$ on the same stage.

Let $v \in \Gamma^N$. In the following, we define $v^{\cal K} \in \Gamma^{N, {\cal K}}$ as follows:
\begin{equation}
v^{\cal K}(S) = v(S) \text{ if } S \in {\cal K}.
\end{equation}

%\begin{equation}
%v_c^k(S) = \left\{
%\begin{array}{ll}
%v(S), & \text{ if } |S| \le k, \\
%v(N) - \sum_{l=2} ^k c_v(l), & \text{ if } S=N, 
%\end{array}
%\right.
%\end{equation}
%for every $k \in \{ 0, 1, \dots, n-1 \}$.

Then, 
a solution for PDGs with costs is represented by $\sigma: \Gamma^N \to \R^{n \times (2^n-n-1)}$ as follows:
\begin{equation}
\sigma(v) = 
\begin{bmatrix}
\sigma_{1,0}(v), & \sigma_{1,1}(v), & \sigma_{1,2}(v), & \dots, & \sigma_{1, 2^n-n-2}(v) \\
\sigma_{2,0}(v), & \sigma_{2,1}(v), & \sigma_{2,2}(v), & \dots, & \sigma_{2, 2^n-n-2}(v) \\
\vdots & \vdots & \dots & \vdots \\
\sigma_{n,0}(v), & \sigma_{n,1}(v), & \sigma_{n,2}(v), & \dots, & \sigma_{n, 2^n-n-2}(v) 
\end{bmatrix}
\end{equation}
%That is, $\sigma(v)$ is represented by a set of $n$-dimensional real numbers for every stage from one to $n-1$. 
where $\sigma_{i,k}(v)$ is a solution of player $i$ after $k$th examination of coalitional worth for every $k \in \{0, 1, \dots, 2^n-n-2 \}$ and $i \in N$.

%Assume a solution $\sigma : \Gamma^N \to \mathbb{R}^{n\times (2^n-n-1) }$.
%Assume that we obtained a solution $\sigma(v)$ for $(N,k)$-PDGs with payments $v$.
%Then,
%let $\Phi^N$ be the set of all solutions.
%Then we define an indicator function $\pi: \Gamma^N \to \{0, 1\}^{2^n-n-2}$ which indicates the exiting rule of 
%the examination of the coalitional worth.
%For each $k \in \{1, \dots, 2^n-n-2\}$, if $\pi_k(v) =0$, then the next examination is performed toward Stage $k$. If $\pi_k(v) =1$,
%then the examination is terminated at Stage $k-1$. 
%If $k=n-1$, then $\pi_k(v) = 1$ for any games. 
In this study, we propose a solution for PDGs with costs and 
an indicator function for them and axiomatize them.

\section{The Shapley Value for Partially Defined Games with Costs and Its Axiomatization}

In this study, we use the Shapley value for partially defined games by \cite{masuya22} to define the Shapley value
for PDGs with costs.
Let $v \in \Gamma^N$. The Shapley value $\Tilde{\phi}:  \Gamma^N \to \R^{n \times (2^n-n-1) }$ is 
defined as follows:

\begin{equation} \label{eq:propose}
\Tilde{\phi}_{i,k}(v)=\sum_{\stackrel{\scriptstyle S \subset N }{S \ni i}}  \frac{d_{v^{\cal K}}(S)}{|S|} + \frac{d_{v^{\cal K}}(N)-\sum_{l =1}^{k} c_v(S_l^v)}{n} \Forall i \in N.
\end{equation}
for all $k \in \{0, 1, \dots, 2^n-n-2  \}$, where ${\cal K}= {\cal L} \cup \{ S_1^v, \dots, S_k^v \}$. 

The proposed value is essentially the same as the value proposed by \cite{masuya22}.
The difference between them is considering the cost to examine the coalitional worth.

We need several definitions before axiomatizing the proposed value. 

\begin{definition}[Partnership]
Let $v \in \Gamma^N$ and $S_k^v \subset N$ for some $k \in \{ 1,2, \dots, 2^n-n-2\}$.
If $v(T) = 0 \Forall T \not\supseteq S_k^v$ holds, then $S_k^v$ is said to be a partnership of $v$.
\end{definition}

A partnership is the coalition in which its members are not included cannot gain payoff.   
A similar concept which is called a p-type coalition is defined by Albizuri et al. \cite{masuya22}. 

A coalition $S$ is said to be a zero-coalition in $v \in \Gamma^{N, {\cal K}}$ if 
$T \subseteq S$ and $T \in {\cal K}$ imply $v(T) = 0$. 

\begin{definition}[p-type coalition]
A non-empty coalition $P \subseteq N$ is said to be a p-type coalition,
in $v \in \Gamma^{N, {\cal K}}$ if for all $S \in {\cal K}$ such that $P \setminus S \neq \emptyset$ it holds:
\begin{enumerate}
\item $S \setminus P \in {\cal K}$ implies $v(S)=v(S\setminus P)$, and
\item $S \setminus P \not\in {\cal K}$ implies $S$ is a zero-coalition. 
\end{enumerate}
\end{definition}
%%%%%%%%%%%%%%%%%%%%%%%%%%%%%%%%%%%%%%%%%%(20251001)%%%%%%%%%%%%%%%%%%%%%%%%%%%%%%%%%%%%%%%%%% 

\begin{definition}[Carrier] \label{df:carrier}
Let $v \in \Gamma^N$ and $S \in { \cal K}$.
If $v^{\cal K}(T) = v^{\cal K}(T \cap S) \Forall T \in {\cal K}$, then $S$ is said to be a carrier of $v^{\cal K}$.
\end{definition}

Definition \ref{df:carrier} is an extension of a carrier for a TU-game which is a well-known concept.

We axiomatize the proposed value $\Tilde{\phi}$. Let $\sigma: \Gamma^N \to \R^{n\times (2^n-n-1)}$. 
%Several definitions are needed before the axiomatization.  

Let $v \in \Gamma^{N}$. 
%In addition, let $v, w \in \Gamma^{N}$ and $d \in \R$. Then we define $(v^k + w^k)(S) = v^k(S) + w^k(S)$
%and $(d v^k)(S) = d \cdot v^k(S)$ for every $S \subseteq N$ such that $1 \le |S| \le k$ or $|S| = n$. 

\begin{axiom}[Efficiency] \label{ax:E}
Let $v \in \Gamma^N$. Then the following holds:
\begin{equation}
\sum_{i \in N} \sigma_{i,k}(v) =v(N) - \sum_{l=1}^k c_v(S_l^v), \text{ for all } k \in \{0, 1, \dots, 2^n-n-2 \}.  
\end{equation}
\end{axiom}

\begin{axiom}[Additivity] \label{ax:ADD}
Let $v_1, v_2 \in \Gamma^N$. Then the following holds:
\begin{equation}
\sigma(v_1 + v_2) = \sigma(v_1) + \sigma(v_2).  
\end{equation}
\end{axiom}

Axiom \ref{ax:ADD} is the adaptation of the axiom of Additivity for the solution of TU-games to PDGs with costs.

\begin{axiom}[Partnership] \label{ax:PARTNER}
Let $v \in \Gamma^N$ and $S_k^v \subset N$ a partnership for some  $k \in \{ 1,2, \dots, 2^n-n-2\}$ of $v$.
Then the following holds:
\begin{equation}
\sigma_{i,k}(v) = \sigma_{j,k}(v) \Forall i,j \in S_k^v.
\end{equation}
\end{axiom}

Axiom \ref{ax:PARTNER} states that all members of the partnership of a game are allocated the same payoff. 
%Axiom \ref{ax:ZERO} states that if a player is a null player, then his allocated payoff is decreased by the cost. 

\begin{axiom}[Carrier] \label{ax:CARRIER}
Let $v \in \Gamma^N$ and $S_k^v \subset N$ a carrier of $v^{\cal K}$ where $k \in \{1, \dots, 2^n-n-2\}$ .
Then, $\sigma_{i,k}(v) = -\frac{\sum_{l=1}^{k} c_v(S_l^v)}{n}$ for all $i \not\in S_k^v$. 
\end{axiom}

Axiom \ref{ax:CARRIER} states that any player who does not belong to a carrier of a game is not allocated a payoff
and only pays costs.

%Axiom \ref{ax:SAME} states that if the degrees of increases with respect to marginal contributions 
%of two players coincide, then the difference of allocated payoffs to them between the two stages are equal.
%%%%%%%%%%%%%%%%%%%%%%%%%%%%%%%%koko(20251002)%%%%%%%%%%%%%%%%%%%%%%%%%%%%%%%%%%%%%%%%%%%%%%%%%%%%%%%%

\begin{axiom}[Fairness of first stage] \label{ax:FAIR}
Let $v \in \Gamma^N$. Then the following holds:
\begin{equation}
\sigma_{i,0}(v) - \sigma_{j,0}(v)= v(i) - v(j), \text{ for all } i,j \in N.  
\end{equation}
\end{axiom}

Axiom \ref{ax:FAIR} is a new axiom and states that the difference of allocated payoffs between two players 
is the difference between their coalitional worth. 

\begin{axiom}[Zero game] \label{ax:ZEROG}
Let $v \in \Gamma^N$ such that $v^{\cal K}(T) =0$ for all $T \in {\cal K} \setminus \{ N \}$
where ${\cal K} = \{ \{ 1 \}, \dots, \{ n \}, S_1^v, \dots, S_k^v, N \}$.  
Then, $\sigma_{i,k}(v) = \frac{v(N)}{n}$ for all $i \in N$.
\end{axiom}

Axiom \ref{ax:ZEROG} is a new axiom and states that all players are allocated the payoff equally
if all worth of known coalitions are zero other than the grand coalition.

%%%%%%%%%%%%%(koko20250903)%%%%%%%%%%%%%%%%%%%%%%%%%%%%%%%%%

\begin{theorem}\label{th:1}
$\Tilde{\phi}$ is the unique function on $\Gamma^N$ that satisfies Axioms \ref{ax:E} through \ref{ax:ZEROG}.
\end{theorem}

\begin{proof}
First, we show that $\Tilde{\phi}$ satisfies Axioms \ref{ax:E} through \ref{ax:ZEROG}.
It is easy to show that $\Tilde{\phi}$ satisfies Axiom \ref{ax:ADD}
by Assumption \ref{as:11}. 
It is easy to verify that $\Tilde{\phi}$ satisfies Axiom \ref{ax:FAIR}. 
It is straightforward to show that $\Tilde{\phi}$ satisfies Axiom \ref{ax:E}. From the definitions of $\Tilde{\phi}$ and a partnership,
$\Tilde{\phi}$ satisfies Axiom \ref{ax:PARTNER}.
From the definitions of $\Tilde{\phi}$ and a carrier, $\Tilde{\phi}$ satisfies Axiom \ref{ax:CARRIER}.
 $\Tilde{\phi}$ satisfies Axiom \ref{ax:ZEROG} from the definition of $\Tilde{\phi}$ and Assumption \ref{as:4}.

Next, we show the uniqueness. Let $\sigma: \Gamma^N \to \R^{n\times (2^n-n-1)}$.
%%%%%%%%%%%%%%%%%%koko(0810)%%%%%%%%%%%%%%%%%%%%%%%%%%%%%%%%%%%%%%%%%%%%%%%%%%%%%%%%%%%%%%%
Let $u_S \in \Gamma^N$ such that $S \subseteq N$ be a unanimity game and $c_S \in \R$.

%Assume that $k$the examination was conducted now. 

(i)  When $|S| = 1$:

Let $i \in N$ and $S =\{ i \}$. Then, from the definition of a unanimity game and Axiom \ref{ax:FAIR},
the following holds:
\begin{equation}\label{eq:th1-1}
\sigma_{j, 0}(c_S u_S) = \sigma_{i, 0}(c_S u_S) -c_S \Forall j \in N \setminus i.
\end{equation}
Moreover, $\sigma_{i, 0}(c_S u_S) + \sigma_{j, 0}(c_S u_S) =c_S$ holds from Axiom \ref{ax:E}.

Therefore, using this equation and equation (\ref{eq:th1-1}), $\sigma_{l,0}(c_S u_S)$ is obtained
uniquely for all $l \in N$. Assume that $\sigma_{i,k}(c_S u_S)$ is obtained uniquely $\Forall i \in N$
$\Forall k \in \{ 0, \dots, m  \}$ for some $m \ge 0$. We show that $\sigma_{i,m+1}(c_S u_S)$ is obtained 
uniquely for all $i \in N$ using the induction with respect to $k$. 
%%%%%%%%%%%%%%%%%%koko(0810)%%%%%%%%%%%%%%%%%%%%%%%%%%%
Let $k=m+1$. If $c_S u_S(S_l) = 0 \Forall l \in \{1, \dots, m+1\}$, then $\{ i \}$ is a carrier of $(c_S u_S)^{\cal K}$.
Therefore, $\sigma_{j, m+1}(c_S u_S)=0 \Forall j \in N \setminus i$ holds from Axiom \ref{ax:CARRIER}
and Assumption \ref{as:4}. Therefore, $\sigma_{i, m+1}(c_S u_S) = c_S$ from Axiom \ref{ax:E}. 

That is, from the induction hypothesis, $\sigma_{l,k}(c_S u_S)$ for all $l \in N$ for all $k \in \{ 0, \dots, m+1  \}$
is obtained uniquely. 

If $c_S u_S(S_{m+1}) = c_S$, then $S_{m+1}$ is a carrier of $(c_S u_S)^{\cal K}$.
Furthermore, $\{ i \}$ is a carrier of  $(c_S u_S)^{\cal K}$ from the definition of $c_S u_S$.
Thus, $\sigma_{j,m+1}(c_S u_S) = -\frac{\sum_{r=1}^{m+1} c_{c_S u_S}(S_r)}{n}$ for all $j \neq i$
holds from Axiom \ref{ax:CARRIER}. 
Therefore, $\sigma_{i,m+1}(c_S u_S) = c_S - \frac{\sum_{r=1}^{m+1} c_{c_S u_S}(S_r)}{n}$
from Axiom \ref{ax:E}. That is, from the induction hypothesis, $\sigma_{l,k}(c_S u_S)$ for all $l \in N$ for all $k \in \{ 0, \dots, m+1  \}$
is obtained uniquely. 

(ii)  When $|S| \ge 2$:

In this case, $\sigma_{i, 0}(c_S u_S) - \sigma_{j, 0}(c_S u_S) = 0 \Forall i,j \in N$ from Axiom \ref{ax:FAIR}.
Therefore, from Axiom \ref{ax:E},  $\sigma_{i, 0}(c_S u_S) = \frac{c_S}{n} \Forall i \in N$.  

Assume that $\sigma_{i,k}(c_S u_S)$ is obtained uniquely $\Forall i \in N$
$\Forall k \in \{ 0, \dots, m  \}$ for some $m \ge 0$. We show that $\sigma_{i,m+1}(c_S u_S)$ is obtained 
uniquely for all $i \in N$ using the induction with respect to $k$. 
Let $k =m+1$. 
If $c_S u_S(S_l) = 0 \Forall l \in \{1, \dots, m+1\}$, then $\sigma_{i,l}(c_S u_S)=\frac{c_S}{n}$ for all $i \in N$
from Axiom \ref{ax:ZEROG}. 

If $c_S u_S(S_{m+1}) = c_S$, then $S_{m+1}$ is a carrier of $(c_S u_S)^{\cal K}$.
Thus, $\sigma_{j,m+1}(c_S u_S)= - \frac{\sum_{r=1}^{m+1} c_{c_S u_S}(S_r)}{n}$ for all $j \not\in S_{m+1}$ from Axiom \ref{ax:CARRIER}.
Moreover, if $S_{m+1}=S$, then $S_{m+1}$ is a partnership of $v$. Thus, $\sigma_{i,m+1}(c_S u_S)=\frac{c_S}{|S_{m+1}|} - \frac{\sum_{r=1}^{m+1} c_{c_S u_S}(S_r)}{n}$
from Axiom \ref{ax:E} and \ref{ax:PARTNER}. If $S_{m+1} \supset S$ and $\{ T \in {\cal K} | S \subseteq T \subset S_{m+1} \} = \emptyset$, then $S_{m+1}$ is a partnership of $v$.
Thus,  $\sigma_{i,l}(c_S u_S)=\frac{c_S}{|S_{m+1}|} - \frac{\sum_{r=1}^{m+1} c_{c_S u_S}(S_r)}{n}$.

If $S_{m+1} \supset S$ and $\{ T \in {\cal K} | S \subseteq T \subset S_{m+1} \} \neq \emptyset$, then  $S$ is a partnership of $v^{\cal K}$ and a carrier of $v^{\cal K}$.
Thus,  $\sigma_{i,m+1}(c_S u_S)=\frac{c_S}{|S|} - \frac{\sum_{r=1}^{m+1} c_{c_S u_S}(S_r)}{n}$ for all $i \in S$ and 
 $\sigma_{i,m+1}(c_S u_S)= - \frac{\sum_{r=1}^{m+1} c_{c_S u_S}(S_r)}{n}$ for all $i \in N \setminus S$ from Axiom \ref{ax:E}, \ref{ax:PARTNER},
and \ref{ax:CARRIER}. 

That is, from the induction hypothesis,  $\sigma_{l,k}(c_S u_S)$ for all $l \in N$ for all $k \in \{ 0, \dots, m+1 \}$ is obtained uniquely.

From Case (i) and (ii), we obtain $\sigma(v)$ uniquely when $v=c_S u_S$ for all $S \subseteq N$.

Furthermore, from Axiom \ref{ax:ADD}, we have:

\begin{equation}
\sigma_{i,k}(v) = \sigma_{i,k}(\sum_{S \subseteq N} c_S u_S) = \sum_{S \subseteq N} \sigma_{i,k}(c_S u_S), 
\end{equation}
for any $v \in \Gamma^N$. 

That is, we obtain $\sigma(v)$ uniquely for all $v \in \Gamma^N$.
\wsq

\end{proof}

\begin{example}\label{ex:1}
Let $N=\{ 1, 2, 3 \}$, $v \in \Gamma^N$, and let $c_v(\{ 1,2 \})=3$, $c_v(\{ 1,3 \})=2$, and $c_v(\{ 2,3 \})=2$. 
Assume that $S_1^v=\{ 1, 3 \}$, $S_2^v=\{ 2,3 \}$, $S_3^v=\{ 1,2 \}$.

The game $v$ is defined as follows:

\begin{align*}
& v(\{ 1 \}) = 5,\ v(\{ 2 \}) = 3,\  v(\{ 3 \}) = 0, \\
& v(\{ 1,2 \}) = 10,\ v(\{ 1,3 \}) = 8,\ v(\{ 2,3 \}) = 5,\ v(\{ 1,2,3 \}) = 20.
\end{align*}

The Harsanyi dividend $d_v$ can be obtained as follows:
\begin{align*}
& d_v(\{ 1 \}) = 5,\ d_v(\{ 2 \}) = 3,\  d_v(\{ 3 \}) = 0, \\
& d_v(\{ 1,2 \}) = 2,\ d_v(\{ 1,3 \}) = 3,\ d_v(\{ 2,3 \}) = 2,\ d_v(\{ 1,2,3 \}) = 5.      
\end{align*}

Then, $\Tilde{\phi}(v)$ is obtained from equation (\ref{eq:propose}) as follows:
\begin{align*}
\Tilde{\phi}_{1,0}(v) = 9,\ \Tilde{\phi}_{2,0}(v) = 7, \ \Tilde{\phi}_{3,0}(v) = 4, \\ 
\Tilde{\phi}_{1,1}(v) = \frac{53}{6}, \Tilde{\phi}_{2,1}(v) = \frac{32}{6}, \Tilde{\phi}_{3,1}(v) = \frac{23}{6}, \\
\Tilde{\phi}_{1,2}(v) = \frac{45}{6}, \Tilde{\phi}_{2,2}(v) = \frac{30}{6}, \Tilde{\phi}_{3,2}(v) = \frac{21}{6}, \\
\Tilde{\phi}_{1,3}(v) = \frac{41}{6}, \Tilde{\phi}_{2,3}(v) = \frac{26}{6}, \Tilde{\phi}_{3,3}(v) = \frac{11}{6}.
\end{align*}
\end{example}

We compare the axioms which are newly considered in this paper with conventional axioms which are well-known. 
We compare Axiom \ref{ax:PARTNER} with the axiom of Symmetric-Partnership which is defined by Albizuri et al.\cite{masuya22}.

\begin{axiom}[Symmetric-Partnership] \label{ax:P-TYPE}
If $P$ is a p-type coalition in $v \in \Gamma^{N, {\cal K}}$, then
\begin{equation}
\sigma_{i,k}(v) = \sigma_{j,k}(v) \Forall i, j \in P \Forall k \in \{1, \dots, 2^n-n-2\}.
\end{equation}
\end{axiom}

Axiom \ref{ax:PARTNER} and \ref{ax:P-TYPE} are similar in some sense. 
We used Axiom \ref{ax:PARTNER} to axiomatize the proposed value since it is more comprehensive.
However, we can axiomatize the proposed value using Axiom \ref{ax:P-TYPE} as shown below. 

\begin{corollary}\label{co:1}
$\Tilde{\phi}$ is the unique function on $\Gamma^N$ that satisfies Axiom \ref{ax:E}, \ref{ax:ADD}, 
\ref{ax:CARRIER}, \ref{ax:FAIR}, \ref{ax:ZEROG}, and \ref{ax:P-TYPE}.
\end{corollary}

%Next, we compare Axiom \ref{ax:FAIR} with the axiom of inessential game which is known as an axiom for TU-games.
%An extension of the axiom of Inessential Game to PDGs with costs can be defined as follows:

%\begin{axiom}[Inessential Game] \label{ax:INESSENTIAL}
%Let $v \in \Gamma^N$ such that $v(N) = \sum_{i \in N} v(i)$
%%%%%%%%%%%%%%%%%%%%%%%%%%%%%%%KOKO(20251002)%%%%%%%%%%%%%%%%%%%%%%%%%%%%%%%%% 
%\begin{equation}
%\sigma_{i,k}(v) = \sigma_{j,k}(v) \text{ for all } k \in \{1, \dots, 2^n-n-2\}.
%\end{equation}
%\end{axiom}

\section{An Indicator Function for the Exit Point of  Examinations of  Coalitional Worth}

Assume a solution $\sigma : \Gamma^N \to \mathbb{R}^{n\times (2^n-n-2)}$ for PDGs with costs.
Let $\Gamma^{N, \alpha} = \{ v \in \Gamma^N | v(N)= \alpha \text{ for some } \alpha \in \R_+\}$.

%Assume that we obtained a solution for $(N,k)$- PDGs $\sigma: \Gamma^N \to \R^{n \times n}$.
Then we define an indicator function $\pi: \Gamma^{N, \alpha} \to \{0, 1\}^{2^n-n-2}$ which indicates the exit point of 
examinations of coalitional worth.
For some $k \in \{1, \dots, 2^n-n-2 \}$, if $\pi_k(v) =0$, then $k$th examination is performed. If $\pi_k(v) =1$,
then the examination is finished. Once the examination is finished, no further examinations are conducted. 
That is, if $\pi_k(v) =1$ for some $k \in \{1, \dots, 2^n-n-3 \}$, then 
$\pi_l(v) =1$ for all $l \in \{ k+1, \dots, 2^n-n-2 \}$. 

We consider an optimal indicator function for examinations of coalitional worth 
using the solution $\Tilde{\phi}$ to determine allocated payoffs for players. 
If we terminate examinations of coalitional worth early, then the cost for examinations is low. However, 
we have to determine the allocated payoffs under uncertain environment. 

We may propose several indicator functions for PDGs with costs.     
In this paper, we propose an indicator function and axiomatize it. 

%\begin{definition}[$k$-null player]
%Let $v \in \Gamma^N$, $i \in N$, and $k \in \{ 1, \dots, n-1 \}$. 
%If $v^k(T) - v^k(T \setminus i) =0$ for all $T \subset N$ such that $|T| \le k$ and $T \ni i$, then
%$i$ is said to be a $k$-null player of $v$. 
%\end{definition}

%The set of $k$-null players is denoted $NP^k$. Moreover, let $|NP^k|=n-k$. 
%If $v^k(N)= v^k(N \setminus NP^k)$ holds, then $NP^k$ is said to be a $k$-null coalition of $v$. 
%%kokomade(20240925)

%Let $\Tilde{\phi}(\bigcup_{k=1}^{n-1} \Gamma^{N,k}) = \{ A \in \R^{n\times (n-1)} | A = [ \Tilde{\phi}(v^1), \dots, \Tilde{\phi}(v^{n-1})] , v \in \Gamma^N \}$, 
%In the rest of this section, we assume that $\sigma = \Tilde{\phi}$.
%Let $NP(v^k)$ the set of null players in $v^k$ under $c_v$ for every $k \in \{ 1, \dots, n-1 \}$.

We define an indicator function $\gamma: \Gamma^{N, \alpha} \to \{ 0, 1 \}^{2^n-n-2}$ as follows:

\begin{equation}\label{eq:if}
\gamma_k(v) = \left\{
\begin{array}{ll}
1, & \text{ if }  v(S_{k-1}^v) \ge v(N) - \sum_{l=1}^{k-1} c_v(S_l^v)  \text{ and } k \neq 1, \\ 
0, & \text{ otherwise. }
\end{array}
\right. 
\end{equation}
for all $k \in \{1, 2, \dots, 2^n-n-2 \}$. 

The proposed function $\gamma$ states that examinations of coalitional worth of a game are terminated if the coalitional worth 
which is newly found is no less than the decreased worth of the grand coalition.   
%Assume that we use $\Tilde{\phi}$ as the allocation rule. 
%From the definition of a null player, if null players appear in a game at some stage, then worth of the grand coalition 
%equals to worth of some coalition whose cardinality is $k$. 
It can be said that we examine coalitional worth as much as possible if we use the proposed stopping rule.

%Here, we give an assumption with respect to an indicator function as follows.

%\begin{assum}\label{as:if}
%Let $v \in \Gamma^N$ and $k \in \{ 2, \dots, n-2 \}$. If $c_v(k) = 0$,
%then $\pi_k(v) = 1$.
%\end{assum}

Next, we axiomatize the proposed indicator function. 
Let $\pi:\Gamma^{N, \alpha} \to \{ 0, 1 \}^{2^n-n-2}$. 
We need several definitions before axiomatizing.

%\begin{definition}[Unity game]
%Let $S \subseteq N$. If $v_S  \in \Gamma^N$
%satisfies the following, then $v_S$ is said to be a unity game:
%\begin{equation}
%v_{S}(T) = \left\{
%\begin{array}{ll}
%1, & \text{ if } T = S, \\
%c, & \text{ if } T = N, \\
%0, & \text{ otherwise. }
%\end{array}
%\right.
%\end{equation}
%where $c$ is a constant satisfying $c > 0$.
%\end{definition}

%\begin{definition} [$k$-no contribution player]
%Let $v \in \Gamma^N$, $k \in\{ 1, \dots, n-1\}$,
%and $i \in N$.  Then $i$ is said to be a $k$-no contribution player of $v$,
%If the following 1 and 2 hold:
%\begin{enumerate}
%  \item $v^k(S)-v^k(S \setminus i) \le 0$ for all $S \subseteq N$ such that $|S| \le k$ and $S \ni i$.
%  \item $v(N)- v(S) - \sum_{l=2}^{k} c_v(l) \le 0 \Forall S \subset N$ such that $|S|= k$ and $S \not\ni i$.
%\end{enumerate} 
%\end{definition}

%Briefly speaking, a $k$-no contribution player is a player whose marginal contribution is no more than zero 
%if the size of the coalition is no more than $k$.  

%Let $\pi:\Gamma^N \to \{ 0, 1 \}^{n-1}$. 

For every $v,w \in \Gamma^{N, \alpha}$, $v+w$ is defined as follows:
\begin{equation}
(v+w)(S) = \left\{
\begin{array}{ll}
\alpha, & \text{ if } S =N, \\
v(S)+w(S) & \text{ otherwise. }
\end{array}
\right.
\end{equation}

Moreover, for every $v,w \in \Gamma^{N, \alpha}$, $v \vee w$ and $v \wedge w$
are defined as follows:

\begin{equation}
(v \vee w)(S) = \max \{ v(S), w(S) \} \Forall S \subseteq N.
\end{equation}
\begin{equation}
(v \wedge w)(S) = \min \{v(S), w(S) \} \Forall S \subseteq N.
\end{equation}

\begin{definition}[Unity game]
Let $v \in \Gamma^{N, \alpha}$ and $S \subset N$ such that $|S| > 1$. If $v_S  \in \Gamma^{N, \alpha}$
satisfies the following, then $v_S$ is said to be a unity game of $v$ and $S$:
\begin{equation}
v_{S}(T) = \left\{
\begin{array}{ll}
v(T), & \text{ if } T \in \{ \{ 1 \}, \dots, \{ n \}, S, N \}, \\
0, & \text{ otherwise. }
\end{array}
\right.
\end{equation}
\end{definition}

%Assume that $\sigma(v^0) = ( \frac{v(N)}{n}, \dots, \frac{v(N)}{n})$.

\begin{axiom}\label{ax:2-1}
Let $v, w \in \Gamma^{N, \alpha}$ and $k \in \{1, \dots, 2^n-n-2 \}$.  
Then $\pi_k(v + w) = \pi_k(v) \vee \pi_k(w)$.
\end{axiom}

Axiom \ref{ax:2-1} is similar to Additivity axiom. 
If boolean values are simply regarded as numbers, then Axiom \ref{ax:2-1}
can be interpreted as Additivity axiom. 
%The difference between them is their operators. 

\begin{axiom}\label{ax:2-2}
Let $v \in \Gamma^{N, \alpha}$ such that $v(N) > 0$, and $k \in \{ 2, \dots, 2^n-n-2 \}$.
If $v(S_{k-1}) = 0$,
then $\pi_k(v) = 0$. 
\end{axiom}

Axiom \ref{ax:2-2} states that if no information is obtained with 
respect to coalitional worth by the examination, then the 
next examination is conducted.

%\begin{axiom}\label{ax:2-5}
%Let $v \in \Gamma^N$ such that $v(N) < 0$. If $\sigma_{i,1}(v) \le 0 \Forall i \in N$, then $\pi_1(v) = 1$. 
%\end{axiom}
%Axiom \ref{ax:2-4} states that if we know the existence of $k$-no contribution player, then the examination is finished.  

%\begin{axiom}\label{ax:2-4}
%Let $v \in \Gamma^N$ and $k \in \{ 1, \dots, n-2 \}$. If $v^k(S) = 0$ for all $S \subset N$ such that $|S| \le k-1$ and
%$v^k(S) \ge 0$ for all $S \subset N$ such that $|S| =k$, then the following holds:
%\begin{equation}
%\sigma_i(v^k) > 0 \text{ for all } i \in N \Leftrightarrow \pi_k(v) =0. 
%\end{equation}
%\end{axiom}

%Axiom \ref{ax:2-4} states that if a $k$-null coalition of a game is identified, 
%then examinations of coalitional worth are finished. 

%\begin{axiom}\label{ax:2-5}
%Let $v \in\Gamma^N$ and $S \subset N$. If $\sum_{i \in S}\sigma_i(v^{|S|}) = v_c^{|S|}(N)$ holds,  
%then $\pi_{|S|}(\sigma_{v}) = 0$.
%\end{axiom}

We obtained the following theorem.

%\begin{prop}\label{pr:if}
%Let $v \in \Gamma^{N, \alpha}$ and $k \in \{ 1,2, \dots, 2^n-n-2 \}$. 
%Then, $\gamma_k(v)=0 \Rightarrow \gamma_k(-v) =1$ holds.
%\end{prop}

%\begin{proof}
%From the definition of $\gamma$, we have:
%\begin{align*}
%& \gamma_k(v) = 0 \\ 
%\Leftrightarrow  & v^k(S) < v^k(N) - \sum_{l=2}^k c_v(l) \Forall S \subset N \text{ such that } |S| =  k \\
%\Leftrightarrow  & -v^k(S) > -v^k(N) + \sum_{l=2}^k c_v(l) \ge -v^k(N) - \sum_{l=2}^k c_{-v}(l)  \\
%\Rightarrow  & \gamma_k(-v) = 1.
%\end{align*}

%Second relation follows from $c_v(l) \ge 0$ for all $l \in \{ 2, \dots, n-1 \}$.
%This completes the proof.
%\wsq
%\end{proof}

\begin{theorem}\label{th:2}
$\gamma$ is the unique function on $\Gamma^{N, \alpha}$ that satisfies Axiom \ref{ax:2-1} and \ref{ax:2-2}.
\end{theorem}

\begin{proof}
First, we show that $\gamma$ satisfies Axiom \ref{ax:2-1} and \ref{ax:2-2}.
It is easy to show that $\gamma$ satisfies Axiom \ref{ax:2-2}.

We show that $\gamma$ satisfies Axiom \ref{ax:2-1}.
Let $v,w \in \Gamma^{N,\alpha}$ and let $k \in \{ 1,2, \dots, 2^n-n-2  \}$.

%%%%%%%%%%%koko(20251029)%%%%%%%%%%%%%%%%%%%%%%%%%%%%%%%%%%%%%%%

%Let $c_v(S_1) \le \dots \le c_v(S_{2^n-n-2})$ and
%$c_w(T_1) \le \dots \le c_w(T_{2^n-n-2})$. 
%Let $c_{v\vee w}(R_1) \le \dots \le c_{v \vee w}(R_{2^n-n-2})$
%such that $c_{v \vee w}(R_l) = \max(c_v(S_l), c_w(T_l))$ for all $l \in \{ 1, \dots, 2^n-n-2 \}$.

(i) When $\gamma_k(v) = 1$ and $\gamma_k(w)=1$:

From the definition of $\gamma$, the following holds:
\begin{equation}\label{eq:th2-1}
v(S_{k-1}^v) \ge v(N) - \sum_{l=1}^{k-1} c_v(S_l^v).
\end{equation}
\begin{equation}\label{eq:th2-2}
w(S_{k-1}^w) \ge w(N) - \sum_{l=1}^{k-1} c_w(S_l^w).
\end{equation}

Thus, we have: 
\begin{align*}
& (v + w)(S_{k-1}^{v+w}) - (v + w)(N) + \sum_{l=1}^{k-1} c_{v + w} (S_l^{v+w}) \\
= & v(S_{k-1}^v) + w(S_{k-1}^w) - v(N) - w(N) + \sum_{l=1}^{k-1} c_{v} (S_l^{v}) 
+ \sum_{l=1}^{k-1} c_{w} (S_l^{w})
\ge 0.
\end{align*}

First equality follows from Assumption \ref{as:11}.
First inequality follows from equation (\ref{eq:th2-1}) and (\ref{eq:th2-2}).
That is, $\gamma_k(v + w) =1$.

(ii) When $\gamma_k(v) = 0$ and $\gamma_k(w)=0$:
%%%%%%%%%%%%%%%%%%%%%%%%%%%%koko(0819)%%%%%%%%%%%%%%%%%%%%%%5

From the definition of $\gamma$, the following holds:
\begin{align*}
& v(S_{k-1}^v) < v(N) - \sum_{l=1}^{k-1} c_v(S_l^v). \\
& w(S_{k-1}^w) < w(N) - \sum_{l=1}^{k-1} c_w(S_l^w).
\end{align*}

Thus, we have: 
\begin{align*}
& (v + w)(S_{k-1}^{v+w}) - (v + w)(N) + \sum_{l=1}^{k-1} c_{v + w} (S_l^{v+w}) \\
= & v(S_{k-1}^v) - v(N) + w(S_{k-1}^w) - w(N)
+\sum_{l=1}^{k-1} c_{v} (S_l^v) + \sum_{l=1}^{k-1} c_{w} (S_l^w) \\
< & 0
\end{align*}

First equality follows from Assumption \ref{as:11}. 
%Third inequality follows from Assumption \ref{as:6}.
%Fourth equality follows from Assumption \ref{as:10}.

That is, $\gamma_k(v + w) =0$.

(iii) When $\gamma_k(v) = 1$ and $\gamma_k(w)=0$:

From the definition of $\gamma$, the following holds:
\begin{align*}
& v(S_{k-1}^v) \ge v(N) - \sum_{l=1}^{k-1} c_v(S_l^v). \\
& w(S_{k-1}^w) < w(N) - \sum_{l=1}^{k-1} c_w(S_l^w).
\end{align*}

If $v(S_{k-1}^v) < w(S_{k-1}^w)$ holds, 
then, from Assumption \ref{as:11}, we have:

\begin{align*}
& -v(S_{k-1}^v) + w(S_{k-1}^w) + v(N) - w(N) - \sum_{l=1}^{k-1} c_{v} (S_l^v) + \sum_{l=1}^{k-1} c_{w} (S_l^w) < 0 \\
\Leftrightarrow & \sum_{l=1}^{k-1} c_{w-v} (S_l^{w-v}) < 0.
\end{align*}

Then, since $c_{w-v}(S_l^{w-v}) \ge 0$ for all  $l \in \{ 1, \dots, k-1 \}$,
this is a contradiction.

If $v(S_{k-1}^v) \ge w(S_{k-1}^w)$ holds, 
then we have:
\begin{align*}
 & v(S_{k-1}^v) - w(S_{k-1}^w) > v(N) - w(N) - \sum_{l=1}^{k-1} c_{v} (S_l^v) 
+ \sum_{l=1}^{k-1} c_{w} (S_l^w) \\
& = \sum_{l=1}^{k-1} c_{w-v} (S_l^{w-v}) \\
& \ge  0.
\end{align*}

First equality follows from Assumption \ref{as:11}.
Here, since $v,w \in \Gamma^{N, \alpha}$ and $v(S_{k-1}^v) \ge w(S_{k-1}^w)$,
$(w-v)(S_{k-1}^{w-v})=0$ for all $k \in \{2, \dots, 2^n-n-2 \}$.

Therefore, $\gamma_k(v+w) = \gamma_k(2v)=\gamma_k(v) = 1$.

Namely, $\gamma_k(v+w)= \gamma_k(v) \vee \gamma_k(w)$ holds.

From cases (i), (ii), and (iii),
$\gamma(v + w) = \gamma(v) \vee \gamma(w)$ holds.

%We can show that $\gamma(v \wedge w) = \gamma(v) \wedge \gamma(w)$ holds
%using the same manner. 

Thus, $\gamma$ satisfies Axiom \ref{ax:2-1}.

Next, we show the uniqueness. Let $\pi: \Gamma^{N, \alpha} \to \{ 0, 1\}^{2^n-n-2}$.
Let $v_S \in \Gamma^{N, \alpha}$ be a unity game of $v \in \Gamma^{N, \alpha}$ such that $S \subset N$. 

From the definition of a unity game, $v_S(T) = 0$ for all $T \neq S$ such that $T \not\in {\cal L}$. 
Therefore, from Axiom \ref{ax:2-2}, $\pi_k(v_S) = 0$ for all $k \in \{ 1,2, \dots, 2^n-n-2 \}$ 
holds. 

Next, for every $v \in \Gamma^{N, \alpha}$, the following relationship holds:

\begin{equation*}
v(S_k) =  \bigvee_{ S \not\in {\cal L}  } v_S(S_k)  .
\end{equation*}
for all $k \in \{ 1,2, \dots, 2^n-n-2 \}$

From Axiom \ref{ax:2-1}, we have:
%%%%koko(0822) %%%%%%%%%%%%%%%%%%%%%%%%
\begin{align*}
\pi_k(v) & =  \pi_k \Big( \sum_{S \not\in {\cal L} } v_S \Big)  \\
& =   \bigvee_{ S \not\in {\cal L} } \pi_k (v_S)  
\end{align*}

First equality follows from the property of unity games. 
Second equality follows from Axiom \ref{ax:2-1}.

This completes the proof.
\wsq

\end{proof}

\begin{example}\label{ex:2}
Using the result of calculation of Example \ref{ex:1}, $\gamma_k(v)$ for all $k \in \{1, 2, 3 \}$ 
is obtained as follows:
\begin{equation}
\gamma_1(v)=0,\ \gamma_2(v)=0,\ \gamma_3(v)=0. 
\end{equation}
\end{example}

We give a further consideration of two axioms which are newly considered above. 
%First, we discuss Axiom \ref{ax:2-1} with Transfer axiom which is first defined by Dubey \cite{dubey75} in the class of 
%simple games. 
In the rest of this section, we consider indicator functions in which Assumption \ref{as:9} is not satisfied.

%\begin{prop}\label{pr:2-1}
%Let $v, w \in \Gamma^N$. 
%If $\pi$ satisfies Axiom \ref{ax:2-1}, then $\pi$ satisfies the following:

%\begin{equation}
%\pi_k(v \vee w) + \pi_k(v \wedge w) = \pi_k(v) + \pi_k(w), \Forall k \in \{ 1, \dots, 2^n-n-1 \}.
%\end{equation}
%\end{prop}

%\begin{proof}
%It is straightforward.
%\wsq
%\end{proof}
%If $\pi$ is a solution for TU-games, then it can be said that Axiom \ref{ax:2-1} is a weaker axiom than Transfer axiom.
%Thus, Axiom \ref{ax:2-1} is a possible axiom to use for characterizing indicator functions. 
%%%%%%%%%%%%%%%%%%%%%%\ref{ax:2-1}%%%%%%%%%%%(20251006)%%%%%%%%%%%%%%%%%%%%%%%%%%%%

%Next, we consider the independence of the axiom \ref{ax:2-1} and \ref{ax:2-2}. 
We define the following two extreme indicator functions $\gamma^A$ and $\gamma^B$:

\begin{equation}
\gamma_k^A(v) = 0 \Forall k \in\{ 1,2, \dots, 2^n-n-2 \}.
\end{equation}

\begin{equation}
\gamma_k^B(v) = \left\{
\begin{array}{ll}
0, & \text{ if }  v(S_{k-1}) = 0 \text{ or } k=1, \\ 
1, & \text{ otherwise. },
\end{array}
\right. 
\end{equation}
for all $k \in \{ 2, \dots, 2^n-n-2 \}$.

$\gamma^A$ is the function which examines all coalitional worth regardless of values of the cost function.
$\gamma^B$ is the function which examines the next coalitional worth if the worth of the coalitions which was examined lastly was zero only.

We axiomatize $\gamma^A$.
We add a new axiom which is defined as follows.

\begin{axiom}\label{ax:2-3}
Let $v \in \Gamma^{N, \alpha}$. Then $\pi(v)$ does not change for all orders ${\cal S}^v$.
\end{axiom}

\begin{theorem}\label{th:3}
$\gamma^A$ is the unique function on $\Gamma^{N, \alpha}$ which satisfies Axiom \ref{ax:2-1},  \ref{ax:2-2}, and \ref{ax:2-3}.
\end{theorem}

\begin{proof}
The proof is very similar to Theorem \ref{th:2}, so the proof was omitted.
\wsq
\end{proof}
 
On the other hand, $\gamma^B$ satisfies Axiom \ref{ax:2-1} and Axiom \ref{ax:2-2} although it does not satisfy Axiom \ref{ax:2-3}. 
Both $\gamma^A$ and $\gamma^B$ are extreme functions in some sense. That is, $\gamma^A$ is used by a third party who does not 
mind values of cost to examine the coalitional worth at all while $\gamma^B$ is used who minds values of cost so much.   
Intuitively, the proposed indicator function seems the function which is well-balanced since it satisfies axioms both $\gamma^A$ and $\gamma^B$ satisfy.

\section{Concluding Remarks and Future Research}

In this paper, we considered a new problem of cooperative game theory that is called a partially defined game
with costs. We proposed and axiomatized the Shapley value for PDGs with costs and an exiting rule which 
indicates when we should stop examinations of coalitional worth. It might be possible to propose exiting rules other 
than the rules we proposed in this paper. For instance, it can be proposed that examinations of coalitional worth are not
performed at all. That is, if we use this exiting rule, then we have to decide allocations of payoffs in the situation that 
only the worth of the grand and singleton coalitions are known. If we use $\Tilde{\phi}$ as the solution for PDGs with costs,
then the allocation rule coincides with the CIS value \cite{driessen91}. Many rules other than this rule and our rule may be proposed 
in future research with their axiom systems.

%\bibliographystyle{plain}
%\bibliography{myrefs}

\end{document}